\documentclass[11pt]{article}
\usepackage{amsthm}
\usepackage{mathtools}
\usepackage{url}
\usepackage{hyperref}
\usepackage{amsmath}
\usepackage{amsfonts}
\usepackage{cite}

\newcommand{\fo}{\mathsf{FO}}
\newcommand{\dynfo}{\ensuremath{\mathsf{DynFO}}}

\newcommand{\C}{\mathcal{C}}
\newcommand{\dync}{\ensuremath{\mathsf{Dyn\C}}}
\newtheorem{theorem}{Theorem}
\newtheorem{lemma}[theorem]{Lemma}
\newtheorem{definition}[theorem]{Definition}
\newtheorem{remark}[theorem]{Remark}
\newtheorem{observation}[theorem]{Observation}
\newtheorem{proposition}[theorem]{Proposition}

\newcommand{\ACz}{\mbox{{\sf AC$^0$}}}
\newcommand{\TCz}{\mbox{{\sf TC$^0$}}}
\newcommand{\TC}{\mbox{{\sf TC}}}

\newcommand{\ACo}{\mbox{{\sf AC$^1$}}}

\newcommand{\NC}{\mbox{{\sf NC}}}
\newcommand{\NCo}{\mbox{{\sf NC$^1$}}}
\newcommand{\Log}{\mbox{{\sf L}}}

\newcommand{\AC}{\mbox{{\sf AC}}}

\newcommand{\DynTCz}{\mbox{{\sf DynTC$^0$}}}
\newcommand{\DynACz}{\mbox{{\sf DynAC$^0$}}}

\newcommand{\MatPow}{\textnormal{\textbf{MatPow}}}
\newcommand{\DynBipMatPow}{\textnormal{\textbf{DynBipMatPow}}}
\newcommand{\DynMatPow}{\textnormal{\textbf{DynMatPow}}}
\newcommand{\Det}{\textnormal{\textbf{Det}}}
\newcommand{\DetPoly}{\textnormal{\textbf{DetPoly}}}
\newcommand{\Interp}{\textnormal{\textbf{Interpolate}}}
\newcommand{\Div}{\textnormal{\textbf{Div}}}

\newcommand{\leqacz}{{\leq}^{\ACz}}

\newcommand {\Q} {{\mathbb{Q}}}
\newcommand {\N} {{\mathbb{N}}}

\title{Dynamic Complexity of Expansion}
\author{ 
Samir Datta\\
\texttt{sdatta@cmi.ac.in}
\and
Anuj Tawari\\
\texttt{atawari@cmi.ac.in}
\and
Yadu Vasudev\\
\texttt{yadu@cse.iitm.ac.in}
}

\begin{document}
\maketitle

\begin{abstract}
Dynamic Complexity was introduced by Immerman and Patnaik
\cite{PatnaikImmerman97} (see also \cite{DongST95}). It has seen a resurgence
of interest in the recent past, see 
\cite{DattaHK14,ZeumeS15,MunozVZ16,BouyerJ17,Zeume17,DKMSZ18,DMVZ18,BarceloRZ18,DMSVZ19,SchmidtSVZK20,DKMTVZ20} for some representative examples.
Use of linear algebra has been a notable feature of some of these
papers. We extend this theme to show that the gap version of spectral 
expansion in bounded degree graphs can be maintained in the class 
$\DynACz$ (also known as $\dynfo$, for domain independent queries) under 
batch changes (insertions and deletions) of 
$O(\frac{\log{n}}{\log{\log{n}}})$ many edges.

The spectral graph theoretic material of this work is based on the paper by
Kale-Seshadri \cite{KaleS11}. Our primary technical contribution is to
maintain up to logarithmic powers of the transition matrix of
a bounded degree undirected graph in $\DynACz$. 
\end{abstract}

\section{Introduction}
Computational complexity conventionally deals  
with problems in which the entire input is given
to begin with and does not change with time.
However, in practice, the input is not always static
and may undergo frequent changes with time.
 For instance, one may want to efficiently update the
result of a query under insertion or deletion
of tuples into a database. In such a scenario, recomputing
the solution from scratch after every update 
may be unnecessarily computation intensive. In this
work, we deal with problems whose solution can be maintained
by one of the simplest possible models of computation:
polynomial size boolean circuits of bounded depth.
The resulting complexity class $\DynACz$ is equivalent to 
Pure SQL in computational power when we think of
graphs (and other structures) encoded as a 
relational database. It is also surprisingly
powerful as witnessed by the result showing that
$\DynACz$ is strong enough to maintain transitive closure 
in directed graphs \cite{DKMSZ18}.
The primary idea in that paper was to reformulate the
the problem in terms of linear algebra. We follow the
same theme to show that expansion in bounded degree graphs
can be maintained in \DynACz.

\subsection{The model of dynamic complexity}
In the \emph{dynamic} (graph) model we start with an empty graph on
a fixed set of vertices. The graph evolves by the insertion/deletion of
a single edge in every time step and some property which can be
periodically queried, has to be maintained
by an algorithm. The dynamic complexity of the algorithm is the
static complexity for each step. If the updates and the queries
can be executed
in a static class $\mathcal{C}$ the dynamic problem is said to belong
to $\dync$. 
In this paper, $\mathcal{C}$ is often a complexity class defined in terms
of bounded depth circuits\footnote{We will have occasion to refer to the
({\sf dlogtime}-)uniform versions of these circuit
classes and we adopt the convention that, whenever unspecified,
we mean the uniform version.} such as $\ACz, \TCz$, where $\ACz$ is the 
class of  polynomial size constant depth circuits with AND and OR gates
of unbounded fan-in; 
$\TCz$ circuits may additionally have MAJORITY gates.
We encourage the reader to refer to any
textbook (e.g. Vollmer \cite{Vollmer}) for precise definitions of the standard circuit complexity classes. 
The model was first introduced by Immerman and Patnaik \cite{PatnaikImmerman97}
(see also Dong, Su, and Topor \cite{DongST95})
who defined the complexity class $\dynfo$ which is essentially 
equivalent\footnote{More precisely, the two classes are equivalent for all 
domain independent queries i.e. queries for which the answer to the query is 
independent of the size of the domain of the structure under question. We will
actually conflate $\dynfo$ with $\dynfo(<,+,\times)$ in which class order
$<$) and the corresponding addition ($+$) and multiplication ($\times$) 
are built in relations. We do so because we need to deal with multiple updates
where the presence of these relations is particularly helpful -- see the 
discussion in \cite{DMVZ18}. We do not need to care about these subtler
 distinctions when we deal with $\DynACz$ in any case.} to the uniform version 
of $\DynACz$. 
The circuit versions $\DynACz, \DynTCz$
were also investigated by Hesse and others \cite{Hes,DattaHK14}.

The archetypal example of a dynamic problem is maintaining
reachability (``is there a directed path from $s$ to $t$''),
in a  digraph. This problem has recently \cite{DKMSZ18} been shown
 to be maintainable in the class $\DynACz$ -
a class where edge insertions, deletions and reachability
queries can be maintained using $\ACz$ circuits. This answers an open question
from \cite{PatnaikImmerman97}. Even more
recently this result has been extended to batch changes of size 
$O(\frac{\log{n}}{\log{\log{n}}})$ (see \cite{DMVZ18}). In this
work, we study expansion testing under batch changes of size similar to above.

\subsection{Expansion testing in dynamic graphs}
In this paper we study the dynamic complexity of checking the expansion of a bounded-degree graph under edge updates. A bounded-degree graph $G$ is an expander if its second-largest eigenvalue $\lambda_G$ is bounded away from $1$. Expanders are a very useful class of graphs with a variety of applications in algorithms and computational complexity, for instance in derandomization. This arises due to the many useful properties of an expander such as the fact that an expander has no small cuts and that random walks on an expander mixes well.

Our aim is to dynamically maintain an approximation of the second largest eigenvalue of a dynamically changing graph in $\DynACz$. We show that for a graph $G$, we can answer if the second largest eigenvalue of the graph is less than a parameter $\alpha$ (meaning that $G$ is a good expander) or if $\lambda_G > \alpha'$ where $\alpha'$ is polynomially related to $\alpha$. The study of a related promise problem of testing expansion was initiated in the sparse model of property testing by Goldreich and Ron \cite{GR11}, and testers for spectral expansion by Kale and Seshadri \cite{KaleS11} and vertex expansion by Czumaj and Sohler
\cite{CS} are known.

Our algorithm is borrowed from the property testing algorithm of \cite{KaleS11} where it is shown that if $\lambda_G \leq \alpha$, then random walks of logarithmic length from every vertex in $G$ will converge to the uniform distribution. On the contrary if $\lambda_G \geq \alpha'$ then this is not the case for at least one vertex in $G$. The key technical contribution in the paper is a method to maintain the logarithmic powers of the normalized adjacency matrix of a dynamic graph when there are few edge modifications.

\subsection{Overview of the Algorithm}
The Kale-Seshadri algorithm \cite{KaleS11} estimates the collision probability 
of several logarithmically long random walks using the lazy transition matrix 
from a small set of randomly chosen vertices. It uses these to
give a probabilistically robust test for the gap version of conductance.
We would like to extend this test to the dynamic setting where the graph
evolves slowly by insertion/deletion of small number of edges. Moreover,
in our dynamic complexity setting the metric to measure the algorithm is not
the sequential time but parallel time using polynomially many 
processors. Thus it suffices to maintain the collision probabilities 
in constant parallel time with polynomially many processors to be able to
solve the gap version of conductance in $\DynACz$.

This brings us to our main result:

\begin{theorem}
\textbf{(Dynamic Expansion test)}
\label{thm:dynamic-expansion-test-intro}
Given the promise that the graph remains
bounded degree (degree at most $d$)
after every round of updates,
\textbf{Expansion testing}\footnote{See Section~\ref{sec:expander-test} for 
a formal definition.} can be
maintained in $\DynACz$ under $O(\frac{\log n}{\log \log n})$
changes.
\end{theorem}

In other words, we need to maintain the generating function of at most logarithmic length 
walks of a transition matrix when the matrix is changed by almost logarithmiclly
many
edges\footnote{$O(\frac{\log{n}}{\log{\log{n}}})$ to be precise -- for us the 
term ``almost logarithmic'' is a shorthand for this.} in one step.
The algorithm is based on a series of reductions, from the above problem
ultimately to two problems -- integer determinant of an almost logarithmic
matrix modulo a small prime and interpolation of a rational polynomial
of polylogarithmic degree.  Each reduction is in the class 
$\ACz$. Moreover if there are errors in the original data the errors do not
increase after a step. On the other hands the entries themselves lengthen
in terms of number of bits. To keep them in control we have to truncate the
entries at every step increasing the error. We can continue to use these values
for a number of steps before the error grows too large. Then we use a matrix
that contains the required generating function computed from scratch.
Unfortunately this ``from scratch'' computation takes logarithmically many
steps. But by this time $O(\frac{\log^2{n}}{\log{\log{n}}})$ changes have
accumulated. Since we deal with almost logarithmic many changes in logarithmically
many steps by working at twice the speed in other words the $\ACz$ circuits
constructed will clear off two batches in one step and thus are of twice
the height. Using this, we catch up with the current change in logarithmically
many steps. 
Hence, we spawn a new circuit at every time step which will be useful 
logarithmically many steps later. 

The crucial reductions are as follows:
\begin{itemize}
\item (Lemma~\ref{lem:genHesse}) 
Dynamically maintaining the aforesaid generating function reduces to 
powering an almost logarithmic matrix of univariate polynomials to logarithmic
powers by adapting (the proof of) a method by Hesse \cite{Hes}.
\item (Lemma~\ref{lem:CayleyHamilton})
Logarithmically powering an almost logarithmically sized matrix reduces
to powering a collection of similar sized matrices but to only an almost
logarithmic power using the Cayley-Hamilton theorem along. This further requires
the computation of the characteristic polynomial via an almost logarithmic sized
determinant and interpolation.
\item  (Lemma~\ref{lem:powToDet})
To compute $M^i$ for $i$ smaller than the size of $M$, we consider  the
power series $(I - zM)^{-1}$ and show that we can use interpolation and 
small determinants (of triangular matrices) to read off the small powers of 
$M$ from it. 
\item (Lemma~\ref{lem:detInterp})
 We reduce rational determinant to integer determinant modulo $p$.
We invoke a known result from \cite{DMVZ18} to place this in $\ACz$.
\end{itemize}
Since interpolation of polylogarithmic degree polynomials is in $\ACz$, this
rounds off the reductions and the outline of the proof of:
\begin{theorem}
\textbf{(Main technical result: informal)}
Let $T$ be an $n\times n$ dynamic transition matrix, in which, there are at most
$O(\frac{\log{n}}{\log{\log{n}}})$ changes in a step.
Then we can maintain in $\DynACz$, a matrix $\tilde{T}$ such that
$|\tilde{T} - T^{\log{n}}| < \frac{1}{n^{\omega(1)}}$.
\label{thm:matrix-pow-approx-intro}
\end{theorem}

\subsection{Motivation}

The conductance of a graph, also referred to as
the uniform sparsest cut in many works, is an important metric
of the graph. Many algorithms have been designed
for approximating the uniform sparsest cut  in the static 
setting  \cite{ACL07, Sherman09, ST13, Madry10, KRV09}.
This naturally raises the question of maintaining an approximate value of the conductance in a dynamic graph 
subject to frequent edge changes.

In the Ph.d Thesis of Goranci \cite{Gor19}, a sequential dynamic incremental algorithm (only edge insertions allowed) 
with polylogarithmic approximation and sublinear worst-case
update time. In \cite{GHTZ20}, the authors give a fully dynamic algorithm (both edge insertions and deletions allowed)
with slightly sublinear approximation and polylogarithmic amortized update time.
 On the other hand, our work gives a fully dynamic
algorithm in a parallel setting.

Another difference is that the algorithm in \cite{Gor19, GHTZ20} outputs
an approximate value of the conductance while our algorithm only
solves the gap version.

There has also been significant related work on investigating the
dynamic complexity of problems like rank, reachability and matching 
under single edge changes \cite{Hes,DattaHK14,DKMSZ18} and under
batch changes \cite{DMVZ18,DKMTVZ20}.

\section{Preliminaries}
We start by putting down a convention we have already been using. We refer
by \emph{almost logarithmic} (in $n$) a function that grows like 
$O(\frac{\log{n}}{\log{\log{n}}})$.
\subsection{Dynamic Complexity}
The primary circuit complexity class we will deal with is $\ACz$, consisting
of languages recognisable by a Boolean circuit family with $\wedge,\vee$-gates 
of unbounded fan-in along with $\neg$-gates of fan-in one where the size of the
circuit is a polynomial in the length of the input and crucially the depth
of the circuit is a constant (independent of the input). Since we are more
interested in providing $\ACz$-upper bounds our circuits will be 
Dlogtime-uniform. There is a close
connection between uniform $\ACz$ and the formal logic class
$\fo$ -- to the extent that \cite{BIS} show that the version $\fo(\leq,+,\times)$ 
 is essentially identical to $\ACz$. We will 
henceforth not distinguish between the two.

The goal of a dynamic program is to answer a given query on an input
graph under changes that insert or delete edges. 
We assume that the number of vertices in the graph are fixed
and initially the number of edges in the graph is zero. Of course, we assume
an encoding for the graph as a string and any natural encoding works.

The complexity of the dynamic program is measured by a complexity class 
$\mathcal{C}$ (such as $\ACz$) and those queries which can be answered 
constitute the class $\mathsf{Dyn}\mathcal{C}$. In other words the (circuit)
class $\mathcal{C}$ can handle each update given some polynomially many stored
bits.

Traditionally the changes
under which this can be done was fixed to one but recently \cite{DMVZ18,DKMTVZ20}
this has been extended to batch changes. In this work we will allow 
nonconstantly many batch changes (of cardinality 
$O(\frac{\log{n}}{\log{\log{n}}})$).

One technique which has proved important for dealing with batch changes is
a form of pipelining suited for circuit/logic classes called 
``muddling'' \cite{DMSVZ19,SVZ}. Suppose we have a static parallel
circuit $\mathcal{A}$ of non-constant depth that can process the input to
a form from where the query is answerable easily and in 
addition we
have a dynamic program $\mathcal{P}$ consisting of constant depth circuits 
that can maintain
the query but the correctness of the results is guaranteed for only a small
number of batches. 
This situation may arise if e.g. the dynamic program uses an approximation
in computing the result and the ensuing errors add up across several steps
making the results useless after a while. On the other hand the static circuit
does precise computation always but takes too much depth.

We need to construct a circuit that is of depth
constant per batch of changes but is promised to work for arbitrarily 
many batches.
The idea is to use a copy of the circuit $\mathcal{A}$ to process the current
input to a form where queries can be answered easily. However, by the time this
happens the input is stale in that $d(\mathcal{A})$ (the depth of the circuit
class) times the batch size many changes are not included. Now we use the 
dynamic program $\mathcal{P}$ to handle a leftover batch and the 
currently arriving batch in one unit of time over the next $d(\mathcal{A})$
time steps. This will allow the program to catch up with the backlog and
allow it to deliver the $2d(\mathcal{A})$ result in the $2d(\mathcal{A})$-th
time step. Since the total depth of the circuit involved in this
is $O(d(\mathcal{A}))$, the average depth remains constant. By starting a
new static circuit at every time step,
that will deliver the then correct result after $2d(\mathcal{A})$-steps,
we are done. To summarise in our particular case (an adapted and modified 
version of the ``muddling'' lemmas from \cite{DMSVZ19,DMVZ18})
we have the following:

\begin{lemma}\label{lem:muddle}
Let $M$ be a matrix with $b = O(\log^2{n})$-bit rational entries. Suppose we have two 
routines available:
\begin{itemize}
\item An algorithm $\mathcal{A}$ that can compute $M^{\log{n}}$ by an
$\ACo$ circuit.
\item A dynamic program $\mathcal{P}$ specified by an $\ACz$-circuit
that can approximately 
maintain $M^{\log{n}}$ under batch changes of size $l
= O(\frac{\log{n}}{\log{\log{n}}})$ for $\Omega(\log{n})$
batches.
\end{itemize}
Then, we have an $\ACz$-circuit that will approximately maintain $M^{\log{n}}$
under batch changes of size $l$ for arbitrarily many batches.
\end{lemma}
\begin{proof}
Suppose the circuit for $\mathcal{A}$ has depth $c_{\mathcal{A}}\log{n}$ and 
the circuit for $\mathcal{P}$ has depth $c_{\mathcal{P}}$. Then we will
show how to construct a circuit $C_t$ at time $t$ of depth 
$d = (c_{\mathcal{A}} + 2c_{\mathcal{P}})\log{n} =c\log{n}$
that will compute the value $M^{\log{n}}$ which is correct at the time 
$t+\log{n}$.
At a time there are $\log{n}$ circuits extant viz. $C_{t-\log{n}+1},C_{t-\log{n}+2},\ldots,C_t$
which will deliver the correct value of $M^{\log{n}}$ at times 
$t+1,t+2,\ldots, t+\log{n}$ respectively. Since the size of each circuit $C_i$ is
polynomial in $n$ so is the size $s_c(n)$ of $c$ layers of $C_i$. 
Thus, we can think that
each layer of the overall circuit consists of $c$ layers of each of 
$C_{t-\log{n}+1},\ldots,C_t$ of total size $s_c(n)\log{n}$ per layer i.e. it is
an $\ACz$ circuit.
\end{proof}

Next we describe the algorithm $\mathcal{A}$ that works in $\ACo$. 
Notice that two $n \times n$ matrices with entries that are rationals with
at most polynomial in $n$ bits each can be multiplied in $\TCz$ 
(see e.g. \cite{HAB,Vollmer}). Hence raising a matrix $A$ to the $\log{n}$-th
power can be done by $\TC$-circuits of depth $O(\log{\log{n}})$ (by repeated
squaring). We also
know that $\TCz$ is a subset of $\NCo$ \cite{Vollmer}
which in turn has $\AC$-circuits of
depth $O(\frac{\log{n}}{\log{\log{n}}})$ (just cut up the circuit into
$\NC$-circuits of depth $\log{\log{n}}$ and expand each subcircuit into a 
DNF-formula of size $2^{2^{\log{\log{N}}}} = n$ -- thus overall we get a depth
reduction by a factor of $\log{\log{n}}$ at the expense of a linear blowup 
in size). Now by substituting these $\AC$-circuit in the $\TC$-circuits of
depth $\log{\log{n}}$ we get an $\ACo$ circuit. Thus we get:
\begin{lemma}\label{lem:staticLogPower}
Let $A$ be an $n \times n$ matrix with rational entries. The entries are 
represented with $n$ bits of precision each. Then computing $A^\ell$,
where $\ell = {O(\log{n})}$, is in $\ACo$.
\end{lemma}

\subsection{Logarithmic space computations}
In this section, we present some basic results about
logarithmic space computations which will be useful 
to us. First, we show that reachability in graphs
can be decided by bounded depth boolean circuits
of subexponential size.

\begin{lemma}
\label{lem:undir-conn-subexp}
Given an input graph $G$ with $|V(G)| = n$ and two fixed vertices
$s$ and $t$, there is a  circuit of depth $2d$ and size
$n^{n^{1/d}}$ which can decide if there is a path from $s$ to $t$
in $G$.
\end{lemma}

See \cite{COST16}, pg. 613 for a proof.

In the following, we denote by $A^{\leq l}$ the words in the language $A$ 
that are of length at most $l$.
\begin{lemma}\label{lem:borrowed}
Suppose $A \in \Log$ is a language. Then for constant $c > 0$,
$A^{\leq \log^c{n}}$ has an $\ACz$ circuit of depth $O(1)$ and size $n^{O(1)}$.
\end{lemma}

 \begin{proof}
 Undirected reachability is $\Log$-hard under first order reductions
 by Cook and McKenzie \cite{CM87}. Hence $A$ reduces to undirected reachability by first order
 reductions. Thus given $N$, we can construct an undirected graph $G_N$ of size
 some $N^k$,
 and two vertices $s,t$ thereof using first order formulas such that
 for all $w$ of length at most $N$, $w \in A$ iff $s,t$ are connected in $G_N$.
 But Lemma~\ref{lem:undir-conn-subexp} tells us that there exists a 
 (very-uniform) $\AC$-circuit of size $N^{kN^{k/d}}$ and depth $2d$
 that determines connectivity in $G_N$. Taking $N = \log^c{n}$, the size of
 the circuit becomes $2^{kc\log{\log{n}}\log^{kc/d}{n}}$. Now pick
 $d = kc + 1$ then the size becomes sublinear in $n$ (because the exponent
 is sublogarithmic).
 \end{proof}

\section{Maintaining expansion in bounded degree graphs}
\label{sec:expander-test}
In this section, we are interested in the problem of
maintaining expansion in a dynamically updating bounded
degree graph. For a degree-bounded graph $G$, let $\lambda_G$ denote the second largest eigenvalue of the normalized adjacency matrix of $G$. First, we define the problem of interest, which we call \textit{Expansion Testing}:

\begin{definition}
(\textbf{Expansion testing})
Given a graph $G$, degree bound $d$, and a parameter $\alpha$, decide whether $\lambda_G \leq \alpha$ or $\lambda_G \geq \alpha'$ where $\alpha' = 1 - (1-\alpha)^2/5000$.
\label{def:exp-test}

\end{definition}

In this section, we aim to prove the following theorem:

\begin{theorem}
\textbf{(Dynamic Expansion test)}
\label{thm:dynamic-expansion-test}
Given the promise that the graph remains
bounded degree (degree at most $d$)
after every round of updates,
\textbf{Expansion testing} can be
maintained in $\DynACz$ under $O(\frac{\log n}{\log \log n})$
changes.

\end{theorem}

Our algorithm is based on Kale and Seshadri's work on testing expansion in the property testing model~\cite{KaleS11}. Our algorithm differs from theirs in that we are working on a dynamic graph, and the major technical challenge is an efficient way to maintain the powers of the normalized adjacency matrix. In this section, we will describe the algorithm and its correctness. In the subsequent sections, we will detail the method to update the power of the normalized adjacency matrix when a small number of entries change.

To prove the theorem, we will first look at the conductance of a graph $G$. For a vertex cut $(S,\overline{S})$ with $|S| \leq n/2$, the conductance of the cut is the probability that one step of the lazy random walk leaves the set $S$. We will denote by $\Phi_G(S)$ the conductance of the cut. Formally, $\Phi_G(S) = \tfrac{|E(S,\overline{S})|}{2d|S|}$. The conductance of the graph $\Phi_G$ is the minimum of $\Phi_G(S)$ over all vertex cuts $(S,\overline{S})$. 
The following inequality between the conductance of a graph and the second largest eigenvalue will be useful in our analysis (see \cite{HooryLW06}).
\begin{align*}
1 - \Phi_G \leq \lambda_G \leq 1 - \frac{\Phi_G^2}{2}.
\end{align*}

 For a $d$-degree-bounded graph $G$, we will think of $G$ as a $2d$-regular graph where each vertex $v \in V$ has $2d -d(u)$ self-loops.

The main idea behind the algorithm in \cite{KaleS11} is to
perform many lazy random walks of length $k = O(\log n)$ from a
fixed vertex $s$ and count the number of pairwise
collisions between the endpoints of these walks.
A \emph{lazy} random walk on a graph from a vertex $v$, chooses a neighbor uniformly at random with probability $1/2d$ and chooses to stay at $v$ with probability $1-d(v)/2d$.
We can compute exactly the probability that two 
different random walks
starting at $s$ collide at their endpoints by
computing $S_s = \sum_{u \in [n]} T^\ell[s][u]
\cdot T^\ell[s][u]$,
where $T$ is a transition matrix of the graph.
Since $T$ is symmetric, the matrix $T^k$ must be a
symmetric matrix. Then $S_s$ is equal to 
the $(s,s)$ entry of the matrix $T^{2k}$.
Hence, it suffices to maintain the $(s,s)$
entry of the matrix $T^{2k}$.

To analyze the lazy random walks in our setting, we will look at transition matrices $T$ such that $T[v,v] = 1 - d(v)/2d$ for every $v\in V$ and  $T[u,v] = d(u)/2d$ for every edge $(u,v) \in G$. Notice that it is equivalent to a random-walk on a $2d$-regular graph, where each vertex $u$ with degree $d(u)$ has $2d-d(u)$ self-loops, and therefore we can use the lemma stated above on the graph.

For a vertex $v \in G$, let $\pi^\ell_v$ denote the distribution over $V$ of lazy random walks of length $\ell$ starting from $v$. The distance of this distribution from the stationary distribution (which is uniform in this case), denoted by $D_\ell(v)$ is given by
\begin{align*}
D_\ell(v)^2 = \sum_{u \in V} \left(\pi^\ell_v(u) - \frac{1}{n} \right)^2 = \sum_{u\in V} \pi^\ell_v(u)^2 - \frac{1}{n}. 
\end{align*}

Observe that $\sum_{u\in V}\pi^\ell_v(u)^2 = T^{2\ell}[v,v]$ as shown in the lemma above.
We now state a technical lemma about the existence of vertex $v$ such that $D_\ell(v)$ is high if the graph has low conductance.

\begin{lemma}[\cite{KaleS11}]
	For a graph $G(V,E)$, let $S \subset V$ be a set of size $s \leq n/2$ such that the cut $(S,\overline{S})$ has conductance less than $\delta$. Then, for any integer $l>0$, there exists a vertex $v\in S$ such that 
	\begin{align*}
	D_l(v) > \frac{1}{2\sqrt{s}} (1 - 4\delta)^l.
	\end{align*}
	\label{lem:ks-low-cond}
\end{lemma}

We can now describe our algorithm for testing expansion. After each update, we use Theorem~\ref{thm:matrix-pow-approx} to obtain the matrix $\tilde{T}$ such that $|\tilde{T} - T^k| \leq 1/n^3$, where $k = \log n/\Phi^2$. Therefore for each $v$, we $\tilde{T}(v,v)$ such that $|\tilde{T}[v,v] - \sum_{u\in V} \pi^\ell_v(u)^2| \leq 1/n^3$. We now test if $\tilde{T}[v,v] \leq \frac{1}{n}\left(1+ \frac{2}{n}\right)$ for each $v \in G$, and reject if this is not the case even for one $v\in G$. The correctness of this algorithm follows from the two lemmas stated below.

\begin{lemma}
	If $\lambda_G \leq \alpha$, then $\tilde{T}[v,v] \leq \frac{1}{n}\left(1+ \frac{2}{n}\right)$ for every $v \in G$.
	\label{lem:exp-yes-case}
\end{lemma}
\begin{proof}
	If $\lambda_G \leq \alpha$, then $\Phi_G \geq 1 - \alpha = \Phi$. Now,
	\begin{align*}
	D_\ell(v)^2 &= \lVert \pi^\ell_v - \frac{1}{n} \rVert_2^2 \leq \frac{1}{n^2}.
	\end{align*}
	Therefore, $T^{2l}[v,v] = \sum_{u\in V} \pi^\ell_v(u)^2 \leq \frac{1}{n}\left( 1 + \frac{1}{n} \right)$. Since $|\tilde{T}[v,v] - T^{2\ell}[v,v]| \leq 1/n^3$, we have $\tilde{T}[v,v] \leq \frac{1}{n}\left(1 + \frac{2}{n}\right)$ for every $v\in V$.
\end{proof}

\begin{lemma}
	If $\lambda_G \geq \alpha'$, then there exists a vertex $v \in G$ such that $\tilde{T}[v,v] > \frac{1}{n}\left(1+ \frac{2}{n}\right)$.
	\label{lem:exp-no-case}
\end{lemma}
\begin{proof}
	If $\lambda_G \geq \alpha'$, then we know that $\Phi_G \leq k\Phi^2$. Therefore, there exists a vertex cut $(S,\overline{S})$ such that $\Phi_G(S) \leq k\Phi^2$. From Lemma~\ref{lem:ks-low-cond} we can conclude that there exists a vertex $v$ such that $D_\ell(v)^2 > \frac{1}{4s} (1 - 4k\Phi^2)^{2\ell} \geq \frac{1}{2n}(1 - 4k\Phi^2)^{2\ell}$. For $\ell=\ln n/8\Phi^2$, and a sufficiently small $k < 1$, we have $D_\ell(v)^2 > \frac{1}{2n^{1+\epsilon}}$ for a small constant $\epsilon > 0$. The collision probability $\sum_{u\in V} \pi^\ell_v(u)^2$ is therefore at least $\frac{1}{n}\left(1 + \frac{1}{n^\epsilon}\right)$. From Theorem~\ref{thm:matrix-pow-approx}, we know that $\tilde{T}[v,v] \geq \frac{1}{n}\left(1 + \frac{1}{2n^\epsilon}\right) > \frac{1}{n}\left(1 + \frac{2}{n}\right)$.
\end{proof}

The key ingredient in the algorithm is a procedure to maintain the logarithmic powers of a weighted adjacency matrix when only a small number of entries change. In the next section we will describe how to do this in $\DynACz$.

\section{Maintaining the logarithmic power of a matrix}
We are given an $n\times n$
 lazy-transition matrix $T$ that varies dynamically with the
batch insertion/deletion of almost logarithmically 
($O(\frac{log{n}}{\log{\log{n}}})$) many edges per time step.
We want to maintain each entry of sum of powers: $\sum_{i=0}^{\log{n}}{(xT)^i}$.
Notice that the exponent $\log{n}$ arises from the Kale-Seshadri expansion-testing
algorithm which needs the probabilities of walks of length
$\log{n}$. On the other hand, the almost logarithmic bound on the small number 
of changes is a consequence of the reductions described below from the
dynamic problem above to ultimately, determinants of small matrices and
interpolation of small degree polynomials. Here interpolation can be done for
degrees up to polylogarithmic but known techniques \cite{DMVZ18}
permit determinants of  at most almost logarithmic size in $\ACz$ yielding
this bottleneck. Another way to view this bottleneck is: while from 
Lemmata~\ref{lem:undir-conn-subexp},~\ref{lem:borrowed}, polylogarithmic
length inputs of languages in $\Log$ (or even $\mathsf{NL}$: see \cite{DMVZ18})
can be decided in $\ACz$, such bounds
are not known for languages reducible to determinants.

\begin{definition}
A $b$-bit rational is a pair consisting of an integer $\alpha$
a natural number $\beta$ such that $|\alpha| < \beta \leq 2^b$.
Its value is $\frac{\alpha}{\beta}$.
By a mild abuse of notation we conflate the pair $(\alpha,\beta)$ with its
value $\frac{\alpha}{\beta}$. 
\end{definition}
\begin{remark}\label{rem:defRemark}
First, notice that every $b$-bit rational is smaller than $1$ by definition.
Second, a $B$-bit approximation $\tilde{r}$ to a rational $r$ may itself 
be a $b$-bit rational for some $b \neq B$. This is because 
the two statements $|r - \tilde{r}| \leq 2^{-B}$ and 
$\tilde{r} = \frac{\alpha}{\beta}$ where $|\alpha| < \beta \leq 2^b$ are 
independent.
\end{remark}

We need some definitions and begin with the definition of a dynamic matrix and
the associated problems of maintaining dynamic matrix powers.
\begin{definition}\label{def:dynamicMatrix}
Let $l \in \N$. 
A matrix $A \in \Q^{n \times n}[x]$ is said to be $(n,d,b,l)$-dynamic if:
\begin{itemize}
\item each coefficient of the polynomials is a $b$-bit rational
\item at every step there is a change in the entries of some $l \times l$ 
submatrix of $A$ to yield a new matrix $A'$. The change matrix 
$\Delta A = A' - A$
\end{itemize}
\end{definition}

\begin{definition}
 $\DynMatPow(n,d,b,k,l)$ is the problem of maintaining the value of each
entry of $\sum_{i=0}^k(xA)^i$ for a $(n,d,b,l)$-dynamic matrix.

Let $\DynBipMatPow(n,d,b,k,l)$ be the special case of $\DynMatPow(n,d,b,k,l)$
where the change matrix $\Delta A$ has a support that is a bipartite graph
with all edges from one bipartition to another.
\end{definition}

Next, we define problems to which the dynamic problems will be reduced to.
We begin with polynomial matrix powering. The last condition in the following
bounding the constant term of entries of the powered matrix 
is a technical one for controlling the error.

\begin{definition}
Let $\MatPow(n,d,b,k)$ be the problem of determining $\sum_{i=0}^k (xA)^k$ for 
a matrix $A \in \mathbb{Q}^{n \times n}[x]$ where all the following hold:
\begin{itemize}
\item the degree of the polynomials is upper bounded by $d$ 
\item each coefficient is a $b$-bit rational
\item the constant term of each polynomial entry is upper bounded by 
${(3n)}^{-1}$
\end{itemize}
\end{definition}
The next group of definitions involve the problems we ultimately reduce the
intermediate matrix powering algorithm to. THese include various determinant
problems, polynomial interpolation and polynomial division.

\begin{definition}
Let $\Det(n,b,v)$ be the problem of computing the value of the determinant of
an $n \times n$ matrix with entries that are 
$b$-bit rationals bounded by $v < 1$ in magnitude. 

Let $\Det_p(n)$ be the problem of computing the value of the determinant of
an $n\times n$ matrix with entries that are from $\mathbb{Z}_p$ for a prime
$p$.

Let $\DetPoly(n,d,b)$ be the problem of computing the value of the determinant 
of an $n \times n$ matrix with entries that are degree $d$ polynomials of
$b$-bit rational coefficients.
\end{definition}

\begin{definition}
Let $\Interp(d,b)$ be the problem of computing the coefficients of a univariate
polynomial of degree $d$, 
where the coefficients are rationals (not necessarily smaller than one)
and where $d+1$ evaluations of the
polynomial on $b$-bit rationals are given.
\end{definition}

\begin{definition}
Let $\Div(n,m,b)$ the problem of computing the quotient of a univariate 
polynomial $g(x)$ of degree $n$ when it is divided by a polynomial $f(x)$
of degree $m$ where both polynomials are monic with other entries being 
$b$-bit rationals.
\end{definition}

In the rest of this section, we will use the following variables consistently:
\begin{itemize}
\item $n$ number of nodes in the graph
\item $l = O(\frac{\log{n}}{\log{\log{n}}})$ the number of changes in one batch
\item $k = O(\log{n})$ the exponent to which we want to raise the transition matrix
\item $b = \log^{O(1)}{n}$ the number of bits in the repersentation
\item $d \leq \log^{O(1)}{n}$ the degree of a polynomial
\end{itemize}

Let us start with the first lemma above:
\begin{lemma}\label{lem:preproc}
$\DynMatPow(n,d,b,k,l)$ reduces to $\DynBipMatPow(2n,d,b,2k,l)$ via
a local\footnote{That is, changing a ``small'' submatrix of the input dynamic 
matrix results in ``small'' small submatrix change in the output matrix of the
reduction. The notion of smallness being almost logarithmic.} $\ACz$-reduction.
\end{lemma}
\begin{proof}
Let $A$ be an $(n,d,b,l)$-dynamic matrix.
Let $B$ be the following $2 \times 2$ block matrix with entries from 
$\mathbb{Q}^{n\times n}$:
\[
\left(
\begin{array}{cc}
0_n & A \\
I_n & 0_n \\
\end{array}
\right).
\]
Here $0_n,I_n$ are respectively the $n\times n$ all zeroes, identity matrices.
Then clearly,
\[
B^{2k} = 
\left(
\begin{array}{cc}
A^k & 0_n \\
0_n & A^k \\
\end{array}
\right)
\]
Notice that:
\[
B' - B = 
\left(
\begin{array}{cc}
0_n & A'-A \\
0_n & 0_n \\
\end{array}
\right),
\]
is a directed bipartite graph with all edges from the first partition of $n$
vertices to the second partition of $n$ vertices, completing the proof.
\end{proof}

\subsection{Generalising Hesse's construction}
\label{subsec:hesse}
Let $G$ be a weighted directed graph with a weight function $w:E \to \mathbb{R}^+$ and weighted adjacency matrix $A$. Let 
$H = H^{(k)}_G(x)$ denote the weighted graph with weighted adjacency matrix 
 $A_H = \sum_{i=0}^k{(xA)^i}$ where $k$ is an integer and
$x$ is a formal (scalar) variable. Let $G'$ be a graph on the vertices
of $G$ differing from $G$ in a ``few'' edges and $A'$ be its adjacency matrix.
Denote by $\Delta A = A' - A$. Notice that $\Delta A$ contains both positive
and negative entries. Let $\Delta_+ A$ be the matrix consisting of the positive 
and $\Delta_- A$ of the negative entries of $\Delta A$. Let $U$ be the
affected vertices i.e. the vertices
on which any of the inserted/deleted edges in $\Delta A$ 
(i.e. the support of the edges whose
adjacency matrices are $\Delta_+ A$ and $-\Delta_- A$) are incident. 
\begin{lemma}\label{lem:genHesse}
Suppose there exists a partition of the affected vertices into two sets 
$U_i,U_o$ such that all inserted and deleted edges are from a vertex in
$U_i$ to a vertex in $U_o$. Consider the matrices $\Delta_{\sigma}$, for $\sigma \in \{+,-\}$, of dimension $|U|+2$,
viewed as a weighted adjacency matrix of a graph on $U \cup \{s,t\}$
where $s,t \in V(G) \setminus U$, and whose entries are defined as
below:

\[
\Delta_\sigma[u,v] = \left\{
\begin{array}{ll}
\sigma w_{uv}x &       \text{if $u \in U_i$ and $v \in U_o$}\\ 
A_H[u,v] & \text{if $u \in U_o$ and $v \in U_i$} \\
A_H[s,v] & \text{if $u = s$ and $v \in U_i$} \\
A_H[u,t] & \text{if $v = t$ and $u \in U_o$} \\
0  & 	   \text{otherwise}\\
\end{array}
\right.
\]

Then the number of $s,t$ walks in $G'$ of length $k \leq \ell$
are given by the coefficient of $x^k$ in $A_H[s,t] + \Delta_\sigma^k[s,t]$.
\end{lemma}
\begin{proof}
	We will prove this separately for the cases $\sigma = +$ and $\sigma= -$. The proof follows the general strategy of Hesse's proof in \cite{Hes}.
	
	When $\sigma=+$, we insert edges into the graph $G$. When new edges are added to $G$, the total number of walks from $s$ to $t$ is the sum of the number of walks that are already present and the new walks due to the insertion of the new edges. Observe that all the $s-t$ walks in the graph on $U \cup \{s,t\}$ must pass through the new edges and every such walk of length $l$ is counted exactly once in $\Delta_+^k[s,t]$.
	
	The more interesting case is when $\sigma=-$, and edges are deleted from $G$. In this case we $A_H[s,t]$ contains all walks from $s$ to $t$ including the deleted edges, and we need to delete only those walks that contain at least one edge that is deleted. The proof follows along the same lines as Hesse's proof when a single edge is deleted. The idea is to show that every walk from $s$ to $t$ of length $l$ is counted exactly once in $A_H[s,t] + \Delta_{-}^k[s,t]$.
	
	Let $P$ be any $s-t$ walk in $G$ that contains edges that are deleted. Firstly, $P$ is counted exactly once in $A_H[s,t]$. Suppose that $r$ of the deleted edges occur in $P$ and the $i^{th}$ edge occurs $k_i$ times. Among the $k_i$ occurrences of the $i^{th}$ edge we can choose $l_i$ occurrences, for each $i$, and this gives a walk where these are the edges from $U_i$ to $U_o$ that we choose in graph on $U \cup \{s,t\}$, and the remaining are counted in the walk from $s$ to $U_i$, $U_i$ to $U_o$ and $U_o$ to $t$. There are $\prod_{i=1}^r\binom{k_i}{l_i}$ such choices, and for each choice the corresponding summand for the walk in $\Delta_-^k[s,t]$ is $(-1)^{l_1+l_2+\cdots+l_r}$. When $l_1=l_2=\ldots=l_r=0$, the walk is counted in $A_H[s,t]$. Therefore, the contribution of the walk $P$ to the sum is given by
	\begin{align*}
	\sum_{l_1 = 0}^{k_1}\sum_{l_2=0}^{k_2}\cdots\sum_{l_r=0}^{k_r} (-1)^{l_1+l_2+\cdots+l_r} \prod_{i=1}^r\binom{k_i}{l_i} &=
	\sum_{l_1 = 0}^{k_1}\sum_{l_2=0}^{k_2}\cdots\sum_{l_r=0}^{k_r} \prod_{i=1}^r (-1)^{l_i} \binom{k_i}{l_i} \\
	&= \prod_{i=1}^r \left( \sum_{l_i=0}^{k_i} (-1)^{l_i} \binom{k_i}{r_i} \right) = 0.
	\end{align*}
	Since the walks that do not pass through the deleted edges never appear in the new graph on $U \cup \{s,t\}$ that we created and are hence counted in $A_H[s,t]$, this completes the proof for the case $\sigma = -$. 
\end{proof}

From the lemma above, we can conclude the following reduction.

\begin{lemma} 
		$\DynBipMatPow(2n,d,b,k,l)$ reduces to \\ $\MatPow(l,d,b,k)$ via an $\ACz$-reduction.
	\label{lem:bip-to-mat}
\end{lemma}
\begin{proof}
	Let $A$ denote the $(2n,d,b,l)$-dynamic matrix such that the support of the changes is a bipartite graph. Let $A_H = \sum_{i=0}^k (xA)^i$. From Lemma~\ref{lem:genHesse}, we know that if $l$ entries of $A$ change, then the $(s,t)$ entry in the new sum, $A_{H'}[s,t]$, can be computed in two steps, first by computing  $A_H[s,t] + \Delta_{-}^k[s,t]$ to obtain $A_{H''}[s,t]$ and then computing $A_{H'}[s,t]$ as $A_{H''}[s,t] + \Delta_{+}^k[s,t]$. The lemma follows from these observations.
\end{proof}

In the following lemma, we analyze the error incurred in
the matrix $A_{H'}$ due to error in the matrix 
$A_H$:

\begin{lemma}
Let $\tilde{A}_{H}$ be a $b$-bit approximation of the matrix $A_H$, then the corresponding matrix $\tilde{A}_{H'}$ obtained from $\tilde{A}_H$ is a $b-1$-bit approximation of $A_{H'}$.
	\label{lem:error-hesse}
\end{lemma}
\begin{proof}
	Let$\tilde{A}_H = A_H - E$ where $E$ denotes an error-matrix with each entry a polynomial with coefficients upper-bounded by $1/2^b$. Each entry in $\tilde{A}_H$ is represented by a $d$-degree polynomial with $b$-bit rational coefficents.
	
	We can compute $\tilde{A}_{H'}[s,t] = \tilde{A}_H[s,t] + \tilde{\Delta}_\sigma^k[s,t]$, where $\tilde{\Delta}_\sigma$ can be constructed from $\tilde{A}_{H'}$. We can write $\tilde{\Delta}_\sigma[s,t] = \Delta_\sigma^k[s,t] - E'[s,t]$, where $E'$ is an error matrix consisting of polynomials of degree at most $d$. We will now show that the coefficients of these polynomials are upper-bounded by $1/2^b$.
	First observe that every entry of $\Delta_\sigma$ is a degree $d$ polynomial with coefficients at most $1/2$. We will bound the term corresponding to $\Delta_\sigma^k$ and the remainder separately. Since each entry of $\Delta_\sigma$ is at most $1/2$, we can bound the first term by $\tfrac{k}{2^k2^b}$. The remainder of the sum can be upper bounded by $\tfrac{k2^k}{2^22^{2b}}$. Therefore, each coefficient of the polynomials of this matrix is bounded by $\tfrac{k}{2^k2^b} + \tfrac{k2^k}{2^22^{2b}} \leq 1/2^b$.
	
	Therefore, we can write $\tilde{A}_{H'}[s,t] = \tilde{A}_H[s,t] + \tilde{\Delta}^k_\sigma[s,t] = A_H[s,t] + \Delta_\sigma^k[s,t] - E''[s,t]$ where $E''$ is an error matrix with each entry bounded by $1/2^{b-1}$. 
\end{proof}

\subsection{From powering to small determinants}
We first need to reduce the exponent from 
logarithmic to almost logarithmic. The following lemma in fact reduces 
it from polylogarithmic to almost logarithmic.
\begin{lemma}\label{lem:CayleyHamilton}
$\MatPow(l,d,b,k)$ $\ACz$-reduces to the conjunction of the following:\\
$\MatPow(l, 0, lb, l)$,
$\DetPoly(l,1,lb + 2k^2d)$,
$\Interp(dk, kl^2b+k^3ld)$ and $\Div(k,l,l^2b + 2k^2ld)$
\end{lemma}
We use a trick (see e.g. \cite{ABD,HealyViola}; notice that the 
treatment is similar but not identical to that in \cite{ABD}
because there we had to power only constant sized matrices)
to reduce large exponents to exponents bounded by the size
of the matrix for any matrix powering problem via the Cayley-Hamilton theorem
(see e.g. Theorem~4, Section 6.3 in Hoffman-Kunze\cite{hoffmanlinear}).
\begin{proof}
Given a matrix 
$M \in \mathbb{Q}^{l \times l}[x]$, let $M_i$ be the value of the polynomial
matrix $M$ with rationals $x_i$ substituted instead of $x$, for 
$i \in \{0,\ldots,dk\}$.
Here $x_0,\ldots,x_{dk}$
are $dk+1$ distinct, sufficiently small
rationals (say $x_i = \frac{i}{(3dk)^2}$).
Let $\chi_{M_i}(z)$ denote its characteristic polynomial $det(zI - M_i)$. We
 write $z^k = q_i(z)\chi_{M_i}(z) + r_i(z)$ for unique polynomials $q_i,r_i$
such that $deg(r_i) < deg(\chi_{M_i}) = l$. Now, 
${M_i}^k = q_i(M_i)\chi_{M_i}(M_i) + r_i(M_i) = r_i(M_i)$. 
Here, the last equality follows from the Cayley-Hamilton theorem
that asserts that $\chi_{M_i}(M_i) = 0_l$. But $r_i(z)$ is a polynomial 
of degree strictly less than the dimension of $M_i$ and each monomial in 
this involves powering $M_i$ to an exponent bounded by $l-1$.
Finally computing $M^k$ reduces to interpolating each entry from the
corresponding entries of $M_i^k$.

Now we analyse this algorithm.
First we evaluate the matrix at $dk+1$ points $x_0, x_1, \ldots, x_{dk}$
where $x_i = \frac{i}{(3dk)^2}$. This yields a matrix $M_i$
whose entries are bounded by $\frac{1}{3k} + \sum_{j=1}^{dk}{i^j(3dk)^{-2j}} < \frac{1}{3k} + \sum_{j=1}^d{(3k)^{-j}} < \frac{1}{3k} + \frac{1}{3k-1} < k^{-1} < {(3l)}^{-1}$ in 
magnitude. 
 
We then compute the characteristic
polynomial of $M_i$.  Notice that the value of $det(zI - M_i)$ is 
a monic polynomial with coefficient of $z^{l-j}$
bounded by $j!{{l}\choose{j}} {(3l)}^{-j} < 1$ for $j > 0$.
Suppose, teh denominator of an entry of $M$ is bounded by $\beta$.
Then The denominator of $M_i$ is bounded by $\beta(3dk)^{2dk}$.
Moreover, the denominator of this coefficient is further bloated 
to at most $\beta^l (3dk)^{2ldk} \leq 2^{lb + 2k^2d}$ (where we use that 
$l\log{3dk} \approx k$).
 Thus this corresponds to an instance of $\DetPoly (l,1, lb + 2k^2d)$.

In the next step, we divide $z^k$ by the characteristic polynomial
of $M_i$, $\chi_{M_i}(z)$. This corresponds to an instance of
$\Div(k,l, l^2b + 2k^2ld)$.

For computing the evaluation of the remainder polynomial on an
$M_i$, we need to
power an $l \times l$ matrix $M_i$ with $lb$-bit rational entries
to exponents bounded by at most $l-1$.
This can be accomplished by
$\MatPow(l,0, lb, l)$ by recalling that the each entry of $M_i$ is bounded
by ${(3l)}^{-1}$.

Finally, we obtain $M^k$ by interpolation. Every entry of $M^k$ is
a polynomial of degree at most $dk$. Every coefficient of
this polynomial is an $kb$-bit rational and moreover the 
evaluation on entries of $r_i(M_i^l)$ are given which are $kl^2b + 2k^3ld$-bit 
i.e. via $\Interp(dk,kl^2b + 2k^3ld)$.
\end{proof}

Next, we reduce almost logarithmic powers of almost logarithmic sized matrices 
to almost logarithmic sized determinants of polynomials.
\begin{lemma}\label{lem:powToDet}
$\MatPow(l,d,b,l)$ $\ACz$-reduces to the conjunctions of \\
$\DetPoly(l,1,lb)$,
$\Det(l+1, lb, (l+1)^{-1})$,
$\Interp(dl,lb)$
\end{lemma}
\begin{proof}
Let $A^{(j)}$ be the univariate polynomial matrix $A = A(x)$ evaluated at 
point $x = x_j$ where
$x_0,\ldots,x_{dl}$ are $dl+1$ distinct rationals say 
$x_j = \frac{j}{(3dl)^2}$.
Consider the infinite power series $p^{(s,t,j)}(z) = (I - zA^{(j)})^{-1}[s,t]$.
$(I - zA^{(j)})^{-1} = \sum_{i=0}^{\infty}{z^i(A^{(j)})^i}$. 
Thus $p^{(s,t,j)}(z)$
is the generating function of $(A^{(j)})^i[s,t]$ parameterised on $i$. 
$p^{(s,t,j)}(z)$ can be also be written, by Cramer's rule (See, for example,
Section 5.4, p. 161, Hoffman-Kunze \cite{hoffmanlinear}). 
as the ratio of two determinants -- the numerator being the determinant of
the $(t,s)$-th minor of $(I - zA^{(j)})$, say $D^{(s,t,j)}(z)$
 and the denominator being the determinant $D^{(j)}(z)$ of $I - zA^{(j)}$.
Thus $p^{(s,t,j)}(z) = \frac{D^{(s,t,j)}(z)}{D^{(j)}(z)}$. In other words,
$D^{(s,t,j)}(z) = p^{(s,t,j)}(z) D^{(j)}(z)$. Now let us compare the 
coefficients of
$z^i$ on both sides\footnote{Here we use the convention that $a_i$ denotes
the coefficient of $z^i$ in $a(z)$, where $a(z)$ is a power series (or  
in particular, a polynomial).}: 
\[
D^{(s,t,j)}_i = \sum_{k=0}^i{p^{(s,t,j)}_k D^{(j)}_{i-k}}
\]
Letting, $i$ run from $0$ to degree of $D^{(j)}$ which is $l$ = dimension of $A$,
we get $l+1$ equations in the $l+1$ unknowns $p^{(s,t,j)}_i$ for 
$i \in \{0,\ldots,l\}, j \in \{0,\ldots,d\}$. Equivalently, this can be 
written as the matrix
equation: $M^{(j)}\pi = d^{(j)}$ where $M^{(j)}$ is an $(l+1)\times (l+1)$ matrix with 
entries $M^{(j)}_{ik} = D^{(j)}_{i-k}$, 
for $0 \leq k \leq i \leq l$ and zero for all other values of $i,k$ lying in
$\{0,\ldots,l\}$. Similarly, $d^{(j)}$ is a vector with entries
$d^{(j)}_k = D^{(s,t,j)}_k$ and $\pi$ the vector with $l+1$ unknowns 
$\pi_k = p^{(s,t,j)}_k$ 
again for $i,k \in \{0,\ldots,l\}$. Notice that specifically in this argument,
indices of matrices/vectors start at $0$ instead of $1$ for convenience.

Next, we show that the matrix $M$ is invertible.
We make the trivial but crucial observation:
\begin{observation}\label{obs:constTerm}
The constant term in $D^{(j)}(z) = det(I - zA^{(j)})$ is $1$.
\end{observation}
This implies that:
\begin{proposition}
$M^{(j)}$ is a lower triangular matrix with all principal diagonal entries
 equal to $1$ hence has determinant $1$.
\end{proposition}
Next we can interpolate the values of $A^i[s,t]$ from the values of 
$(A^{(j)})^i[s,t] = [z^i]p^{(s,t,j)}$ for 
$i \in \{0,\ldots,l\}, j \in \{0,\ldots,\frac{dl}{(3dl)^2 }\}$.`

Now we analyse this algorithm.
First, we evaluate the matrix at $dl+1$ distinct rationals.
Each entry of the $j$-th matrix is now bounded by 
$\frac{1}{3l} + \sum_{i=1}^{dl}{j^i(3dl)^{-2i}}< {(l+1)}^{-1}$. 
We then compute determinant
of the matrix $I - zA^{(j)}$ where $A \in \Q^{l \times l}$. The number of bits 
in the denominators are less than $2^{bl}$ and moroever the values of 
coefficient of $z^j$ is less than $j!{l\choose{j}}{(l+1)}^{-j} < 1$ if $j > 0$.
 Thus the coefficents are $lb$-bit rationals.
Hence, computing the determinant of $I - zA^{(j)}$ corresponds to an instance of
$\DetPoly(l,1,lb)$.

The next step is to compute the inverse of a $(l+1) \times (l+1)$
matrix $M^{(j)}$ above.
 The $(a,b)$ entry of $M^{(j)}$ is either zero
or equals the coefficient of $z^{a-b}$ in a cofactor of
$(I -zA^{(j)})$. 
By a logic similar to that used in the proof of 
Lemma~\ref{lem:CayleyHamilton} these are all $lb$-bit rationals.
Further, a crude upper bound on their values is 
$(l+1)^{-1}$ as above.
Thus we get an instance of $\Det(l+1,lb,(l+1)^{-1})$.

Finally, we find the matrix $A^i$ from the values of $(A^{(j)})^i$.
This corresponds to an instance of 
$\Interp (dl, lb)$.
\end{proof}

\subsection{Working with small determinants}
First we reduce the problem of computing almost logarithmic determinants of 
small polynomials to computing almost logarithmic determinants over small
rationals.

\begin{lemma}\label{lem:detInterp}
$\DetPoly(l,d,b)$ $\ACz$-reduces to the conjunction of
\\ $\Interp(dl,lb),\Det(l,lb,(l+1)^{-1})$.
\end{lemma}
\begin{proof}
Here, we need to compute the determinant of an $l \times l$
matrix with each entry a degree $\leq d$ polynomial with $b$ bit
coefficients. Clearly, the determinant is a degree at most $ld$
polynomial. So, we plug in
$ld + 1$ different values $x_0, x_1, \ldots, x_{ld}$, 
where $x_i = \frac{i}{(3ld)^2}$ into the
determinant polynomial. This
yields a determinant with entries bounded by 
$(l+1)^{-1}$ in magnitude as in the proof of Lemma~\ref{lem:CayleyHamilton}.

The next step is to interpolate a degree at most $ld$ polynomial.
The coefficient of $x^m$ where $m \leq ld$ is bounded by
$1$ as in the previous lemmas, and the number of bits is at most
$lb$ as well.

\end{proof}

Then we show how to compute small rational determinants by using
Chinese Remaindering and the computation of determinants over small fields.
\begin{lemma}\label{lem:rationalDMVZ}
For every $c > 0$,
$\Det(\frac{\log{n}}{\log{\log{n}}}, \log^c{n},v) \in \ACz$
\end{lemma}
\begin{proof}
Let $b = \log^c n$.
The basic idea is to use the Chinese Remainder Theorem, CRT
 (see e.g. \cite{HAB})
 with prime moduli that are of magnitude at most $\log^{c+1}{n}$ to obtain
the determinant of $2^bA$ which is an integer matrix (since the entries are
$b$-bit rationals of magnitude at most $2^b$). For $n^{O(1)}$ primes this
problem is solvable in $\TCz$ by \cite{HAB} and hence in $\Log$. Thus, by
Lemma~\ref{lem:borrowed} it is in $\ACz$ for primes of magnitude 
polylogarithmic. We of course need to compute the determinants modulo the
small primes for which we use Lemma~\ref{lem:DMVZ}.
\end{proof}

Notice that in the definition of $\Det$, the third argument,
that is $v$, is needed for error analysis which is done in the  following lemma. For our purposes, $v \leq \frac{1}{l+1}$. 

\begin{lemma}\label{lem:app:errorDet}
Let $A$ be a $l \times l$ matrix with entries that are $b$-bit rationals
smaller than $l^{-1}$.
Let $\tilde{A}$ be a $l\times l$ matrix each of whose entries is a 
$B$-bit approximation to the corresponding entry of $A$. Then 
assuming $B = \Omega(l^2)$,
$det(\tilde{A})$ is a $B$-bit approximation to $det(A)$.
\end{lemma}
\begin{proof}
Difference between corresponding monomials in the two determinants is
easily seen to be upper bounded by $2^{-B}l^{-(l-1)}$ in magnitude.
Notice that the assumption
$B = \Omega(l^2)$ tacitly implies that we can neglect all monomials that
that contain more than one term of magnitude $2^{-B}$ and just need to 
consider the terms that consist of exactly one $2^{-B}$ and the rest being
the actual entries.
Hence the (signed) sum over all monomial (differences) is upper bounded by 
$2^{-B}$ in magnitude.
\end{proof}

\begin{lemma}(Paraphrased from Theorem~8 \cite{DMVZ18})\label{lem:DMVZ}
If $p \in O(n^c)$ is a prime then,
$\Det_p(\frac{\log{n}}{\log{\log{n}}}) \in \ACz$.
\end{lemma}

\subsection{The complexity of polynomial division}
We use a slight modification of the Kung-Sieveking algorithm as described in 
\cite{ABD,HealyViola}. The algorithm in \cite{ABD} worked over finite fields 
while here we apply it to divide polynomials of small heights and degrees over
rationals. The algorithm and its proof of correctness
follows Lemma~7 from \cite{ABD} in verbatim. 
We reproduce the relevant part for completeness
 (with minor emendments to accommodate for the characteristic):
\begin{lemma}\label{lem:KungSieveking}
Let  $g(x)$ of degree $n$ and $f(x)$ of degree $m$ be monic univariate
polynomials over $\mathbb{Q}[x]$, such that
$g(x) = q(x)f(x) + r(x)$
for some polynomials $q(x)$ of degree $(n - m)$ and $r(x)$ of degree $(m - 1)$.
Then, given the coefficients of $g$ and $f$, the coefficients of $r$ can
be computed in $\TCz$.
In other words $\Div(n,m,b) \in \TCz$ if $m < n$ and $b = n^{O(1)}$.
\end{lemma}
\begin{proof} Let $f(x) = \sum_{i=0}^m{ a_i x^i}$, 
$g(x) = \sum_{i=0}^n{b_i x^i}$,
$r(x) = \sum_{i=0}^{m-1}{r_i x^i}$ and $q(x) = \sum_{i=0}^{n-m}{q_i x^i}$. 
Since $f, g$ are monic, we have $a_m = b_n = 1$.
Denote by $f_R(x), g_R(x), r_R(x)$ and $q_R(x)$ respectively the polynomial with the $i$-th
coefficient $a_{m-i}, b_{n-i}, r_{m-i-1}$ and $q_{n-m-i}$ respectively. 
Then note that
$x^m f(1/x) = f_R(x)$, $x^n g(1/x) = g_R(x)$, $x^{n-m} q(1/x) = q_R(x)$ and $x^{m-1} r(1/x) = r_R(x)$.

We use the Kung-Sieveking algorithm (as implemented in \cite{ABD}). 
The algorithm is as follows:
\begin{enumerate}
 \item Compute $\tilde{f}_R(x) = \sum_{i=0}^{n-m} (1 - f_R(x))^i$ via interpolation.
 \item Compute $h(x) = \tilde{f}_R(x) g_R(x) = c_0 + c_1x + \ldots + c_{d(n-m) + n}x^{m(n-m) + n}$.
  from which the coefficients of $q(x)$ can be obtained as $q_i = c_{m(n-m) + n - i}$.
 \item Compute $r(x) = g(x) - q(x)f(x)$.
\end{enumerate}
The proof of correctness of the algorithm is identical to that in \cite{ABD}.
The proof of the lemma is immediate because polynomial product is in $\TCz$
from \cite{HAB}.
\end{proof}

\begin{lemma}\label{lem:div}
$\Div(k,l,b)$ $\ACz$-reduces to  $\Interp(kl, kb)$
 \end{lemma}
 \begin{proof}
 In the first step of the algorithm from the proof of 
Lemma~\ref{lem:KungSieveking},
we need to interpolate a polynomial of degree at most $(k-l)l$.
Also, the coefficients of the polynomial are rationals with rationals that
are $(k-l)b$ bits long. Notice that
we do not require the coefficients of the polynomial to be smaller than $1$.
 \end{proof}
 
\subsection{Interpolation error analysis}

The following lemma shows that no precision is lost 
during each call to $\Interp$.

\begin{lemma} \label{lem:app:errorInterp}
Let $f(z)$ be a polynomial of degree $d$ with entries that are
rationals not necessarily smaller than $1$. Suppose, $z_i = \frac{i}{(3d)^2}$ for 
$i \in \{0,\ldots,d\}$ are $d+1$ values. If we know $B$-bit approximations
$\tilde{f}_i$ to the values $f(z_i)$, then the interpolant of these 
values is a function $\tilde{f}$ whose coefficients are at least 
$B$-bit approximations of the corresponding coefficients of $f$.
\end{lemma}
\begin{proof}
(Lagrange) interpolation can be viewed as computing $V^{-1}F$ where $V$ is
 a $(d+1)\times (d+1)$ Vandermonde matrix
\footnote{We assume that the indices of the Vandermonde matrix run in
$\{0,\ldots,d\}$ and that $0^0 = 1$ for convenience.},
such that $V_{ij} = z_i^j$ while $F_i = f(z_i)$ are entries of a column vector.
 The determinant
of the Vandermonde matrix is $\prod_{0\leq j < i \leq d}(z_i- z_j)$. This equals
$\prod_{0\leq j<i\leq d}{\frac{i-j}{(3d)^2}} = \left(\prod_{i=1}^d{i!}\right)\frac{1}{(3d)^{d(d-1)}}$.
On the other hand, the various co-factors are upper bounded in magnitude 
by $d! \prod_{i = 1}^d{z_i^{d-i+1}} = d! \prod_{i=1}^d\frac{i^{d-i+1}}{(3d)^{2(d-i+1)}} = d! \frac{1^d 2^{d-1} \ldots (d-1)^2 d}{(3d)^{2\sum_{i=1}^d{i}}} = d! \frac{\prod_{i=1}^d{i!}}{(3d)^{d(d+1)}}$ 
by considering the monomial with the largest magnitude. 

Thus an entry of the inverse i.e. the 
ratio of a co-factor and the determinant is upper bounded by:
$\frac{d! (3d)^{d(d-1)}}{(3d)^{d(d+1)}} = \frac{d!}{(3d)^{2d}} < 1$.

Hence the coefficients of $V^{-1}(F-\tilde{f})$ (where $\tilde{f}$ is the
column vector with entries $\tilde{f}_i$)
are bounded by $2^{-B}$ completing the proof. Notice that we do not use the 
magnitude of $f(z_i)$ or of $\tilde{f}_i$ in the proof but only that their
difference is small.
\end{proof}  

\subsection{Putting it together with error analysis}
We now reach the main theorem of this Section:

\begin{theorem}
Let $T$ be an $(n,\log{n},\log^2{n},\frac{\log{n}}{\log{\log{n}}})$-dynamic adjacency matrix.
Then we can maintain in $\DynACz$, a matrix $\tilde{T}$ such that 
$|\tilde{T} - T^{\log{n}}| < \frac{1}{n^{\omega(1)}}$.
\label{thm:matrix-pow-approx}
\end{theorem}
\begin{proof}
We use the reductions presented in Lemmas~\ref{lem:preproc},~\ref{lem:bip-to-mat},
~\ref{lem:CayleyHamilton},~\ref{lem:powToDet},~\ref{lem:detInterp} and \ref{lem:div} to prove the result.
\begin{eqnarray*}
\DynMatPow(n,d,b,k,l)           &         & \\
(Lemma~\ref{lem:preproc})	& \leqacz & \DynBipMatPow(2n,d,b,2k,l) \\
(Lemma~\ref{lem:bip-to-mat})      & \leqacz & \MatPow(l, d, b, 2k) \\
(Lemma~\ref{lem:CayleyHamilton})& \leqacz & \MatPow(l, 0, lb, l) \wedge \\
		                &         & \DetPoly(l,1,lb+8k^2d) \wedge \\
				&         & \Interp(2dk,2kl^2b + 8k^3ld) \wedge  \\
				&         & \Div(2k,l,l^2b + 8k^2ld) \\
(Lemma~\ref{lem:powToDet})	& \leqacz & \DetPoly(l,1,lb) \wedge \\
				&         & \Det(l+1,lb,(l+1)^{-1}) \wedge\\
				&         & \Interp(0,lb) \wedge \\
		                &         & \DetPoly(l,1,lb+8k^2d) \wedge\\
				&         & \Interp(2dk,2kl^2b+8k^3ld) \wedge\\
				& 	  & \Div(2k,l,l^2b+8k^2ld) \\
(Lemma~\ref{lem:detInterp})     & \leqacz & \Interp(dl,ldb) \wedge\\
				&      	  & \Det(l,ldb,(l+1)^{-1}) \wedge \\
		      		&	  & \Interp(2dk,2kl^2b + 8k^3ld)\wedge\\
(Lemma~\ref{lem:div})	        &	  & \Interp(2kl,2kl^2b+8k^3ld) \\
		                & \equiv  & \Interp(2\log^2{n}, \\
				&         & \frac{10\log^{5}{n}}{(\log{\log{n}})}) \wedge \\
		                &	  & \Det(\frac{\log{n}}{\log{\log{n}}}, 
				       \frac{\log^{4}{n}}{(\log{\log{n}})},
				       (\frac{\log{n}}{\log{\log{n}}}+1)^{-1}) \\
\end{eqnarray*}

Each $\DynMatPow$ call boils down to a number of $\Det,\Interp$ calls as above.

Though there is no loss of precision in each call to $\Det$
(Lemma~\ref{lem:app:errorDet}) and $\Interp$ (Lemma~\ref{lem:app:errorInterp}),
in Lemma~\ref{lem:bip-to-mat} we lose $O(1)$-bits
of precision (See Lemma~\ref{lem:error-hesse}).  However, the length of
the bit representation grows by a factor of $9\frac{\log^3{n}}{\log{\log{n}}}$ 
at every batch. Thus to 
keep the number of bits under control we need to truncate the matrix at 
$\log^2{n}$-bits again so that now the powered matrix is now a $\log^2{n}-O(1)$-bit
approximation. This $O(1)$ will deteriorate at every step so that we can 
afford to perform at least $\Omega(\log{n})$ steps before we recompute the 
results from scratch, that is, do muddling.

We can do muddling by invoking Lemma~\ref{lem:muddle} where we pick 
$\mathcal{A}$ to be the algorithm from Lemma~\ref{lem:staticLogPower} and with
the above sequence of reductions as the dynamic program $\mathcal{P}$ for 
handling a batch (or actually two batches -- one old and one new) of changes.
\end{proof}

\section{Conclusion}

In this paper we solve a gap version of the expansion testing problem, wherein we want to test if the expansion is greater than $\alpha$ or less than a $\alpha'$. The dependence of $\alpha'$ on $\alpha$ in this paper is due to the approach of using random walks and testing the conductance. It is natural to ask if there is an alternate method leading to a better dependence. A more natural question is whether we can maintain an approximation of the second largest eigenvalue of a dynamic graph. 

An alternative direction of work would be to improve on the number of updates 
allowed per round. In this paper, we show how to test expansion when almost
logarithmic ($O(\frac{\log{n}}{\log{\log{n}}})$) changes are allowed per round. 
The largest determinant that we can compute in $\ACz$ is at most of an almost 
logarithmic size and is the bottle neck that prevents us from improving this 
bound. We know of another way (obtained via careful adaptation of a proof in \cite{Nisan94}) to approximate 
the powers of the transition matrix when $\log^{O(1)} n$ changes are allowed 
per round. Unfortunately, we don't get a strong enough approximation that 
leads to an algorithm for approximating conductance.

\section*{Acknowledgements} SD would like to thank Anish Mukherjee, 
Nils Vortmeier and Thomas Zeume for many interesting and illuminating 
conversations over the years and in particular for discussions that ultimately
crystallized into Lemma~\ref{lem:genHesse}. We would like to thank Eric Allender for clarification regarding previous work.
 SD was partially funded by a grant from Infosys foundation and SERB-MATRICS grant MTR/2017/000480. AT was partially funded by a grant from Infosys foundation.

\bibliographystyle{alpha}
\bibliography{ev}

\end{document}